\documentclass[aps,pra,twocolumn,amssymb,amsmath,superscriptaddress]{revtex4-1}

\usepackage{amsthm}
\usepackage{graphicx}

\newtheorem{lemma}{Lemma}

\begin{document}

\title{Analysis of a rate-adaptive reconciliation protocol and the effect of the leakage on the secret key rate}
\author{David~Elkouss}
\email{delkouss@mat.ucm.es}
\affiliation{Departamento de Analisis Matem\'{a}tico and Instituto de Matem\'{a}tica Interdisciplinar, Universidad Complutense de Madrid, 28040, Madrid, Spain}
\author{Jesus~Martinez-Mateo}
\author{Vicente~Martin}
\email{\{jmartinez, vicente\}@fi.upm.es}
\affiliation{Facultad de Inform\'{a}tica, Universidad Polit\'{e}cnica de Madrid (UPM) Campus de Montegancedo, 28660 Boadilla del Monte (Madrid), Spain}

\begin{abstract}

Quantum key distribution performs the trick of growing a secret key in two distant places connected by a quantum channel. The main reason is that the legitimate users can bound the information gathered by the eavesdropper. In practical systems, whether because of finite resources or external conditions, the quantum channel is subject to fluctuations. A rate adaptive information reconciliation protocol, that adapts to the changes in the communication channel, is then required to minimize the leakage of information in the classical postprocessing. 

We consider here the leakage of a rate-adaptive information reconciliation protocol. The length of the exchanged messages is larger than that of an optimal protocol; however, we prove that the min-entropy reduction is limited. The simulation results, both on the asymptotic and in the finite-length regime, show that this protocol allows to increase the amount of distillable secret key. 

\end{abstract}

\maketitle

\section{Introduction}

Claude E. Shannon published his seminal ``A mathematical theory of communications''~\cite{Shannon_48} in 1948 after eight years of intermittent work~\cite{Gallager_01}. The paper meant the birth of communications and coding theory. Shannon did not only establish the frame under which communications systems could be studied and compared, he also proved their fundamental limits, i.e. the limiting rates for data compression and reliable transmission through noisy channels. This second result was specially surprising since there was no certainty that reliable transmission with a positive rate was even possible~\cite{MacKay_03}.

A year later, in 1949, Shannon's ``Communication theory of secrecy systems''~\cite{Shannon_49} came to light. In words of Robert Gallager ``Shannon's cryptography work can be viewed as changing cryptography from an art to a science''~\cite{Gallager_01}. Shannon successfully applied the tools developed in~\cite{Shannon_48} to the problem of transmitting confidential messages through public channels. His main conclusion is that a message from a set of messages sent through a public channel can be obfuscated into a cypher-text with the help of a secret key in such a way that the number of possible originating messages is the whole set of messages, that is, the cypher-text leaks no information to a possible eavesdropper. The condition for this to happen is that the number of secret keys is equal or greater than the number of messages. This condition only applies to eavesdroppers with unbounded resources, if we limit the storage or computing capability of the eavesdropper secret communications are possible without fulfilling the condition. It is evident that computing power resources that today might be considered as out of reach might become available in the near future. There is an implicit risk in assuming that an eavesdropper is limited in any way beyond the fundamental limits that physics impose her, therefore the interest in establishing the scenarios in which some kind of security can be achieved without any assumption is self-evident.

The distribution of secret keys or SKD is a problem closely related to confidential communications. Two parties sharing a secret key can communicate privately through a channel in the conditions discussed in the previous paragraph. We can then study the problem of secret key sharing as a way to achieve confidential communications. The main idea is that two distant parties can agree in a secret key if they have access to a shared source of randomness. The randomness source can take many incarnations, e.g. in the form of a source received from a trusted party or in the form of a noisy channel~\cite{Maurer_93}.

In most of the SKD scenarios the legitimate parties obtain instances of correlated sources which means that they obtain similar but not identical strings. It is then assumed that there is an authentic though otherwise public channel available to all parties---including the eavesdropper. The legitimate parties can exchange additional information through this channel in order to reconcile their strings. They can do so by revealing some information about them, for instance the parities of carefully chosen positions. This process is known as information reconciliation~\cite{Brassard_94}. It is not hard to see that the information exchanged through the public channel reduces the uncertainty that the eavesdropper has on the strings of the legitimate parties. Thus, a reduction in the leakage due to information reconciliation allows to increase the amount of distillable secret key. A second step known as privacy amplification is then needed~\cite{Bennett_95}. In the privacy amplification step the legitimate parties agree on a secret but shorter key of which the eavesdropper has a negligible amount of information.

These mathematical models can have a real, i.e. physical correspondence. One such a model is a physical fiber carrying single photons randomly polarized in one of  two non-orthogonal basis~\cite{Bennett_84}. Quantum key distribution (QKD) is probably the main practical application of SKD. In a QKD protocol \cite{Bennett_84, Gisin_02, Scarani_09}, two legitimate parties, Alice and Bob, aim at sharing an information theoretic secret key, even in the presence of an eavesdropper Eve. In the quantum part of such a protocol, Alice and Bob exchange quantum signals, e.g. single photons, which carry classical information. For instance, Alice encodes a classical bit onto the polarization or the phase of a photon and sends this photon to Bob who measures it. In any realistic implementation of a QKD protocol, the strings obtained after the exchange of the quantum signals suffer discrepancies mainly due to losses in the channel and noise in Bob's detectors but which are conservatively attributed to the action of an eavesdropper. Therefore, any QKD protocol must include the classical post-processing steps described above in order to extract a secret key from the correlated strings.

The channel connecting Alice and Bob in a real system may substantially vary over time. The motivation of this work is to analyze the $sp$-protocol \cite{Elkouss_11}, an information reconciliation protocol that adapts to this channel variations. We had previously showed that in a classical repetition scenario (i.e. with classical attackers and independent, identically distributed sources) its reconciliation efficiency is only limited by the quality of the error correcting code used to implement the protocol \cite{Elkouss_10}. We consider here the leakage of the $sp$-protocol with a quantum eavesdropper, both in the asymptotic and in the finite-length regime, and its impact on the amount of distillable secret key.

\section{Preliminaries and Notation}
\label{sec:notation}

Let $X$ be a discrete random variable taking values in the finite alphabet $\mathcal X$. The Shannon entropy~\cite{Shannon_48}, min-entropy and max-entropy~\cite{Renyi_60} of $X$ are respectively defined by:

\begin{equation}
\label{eq:entr}
H(X) = - \sum_{x \in {\mathcal X}} p_X(x) \log p_X(x)
\end{equation}

\begin{equation}
\label{eq:entrinfity}
H_{\infty}(X) = \min_{x \in {\mathcal X}} \left( - \log  p_X(x) \right)
\end{equation}

\begin{equation}
\label{eq:entr0}
H_0(X) = \log | x \in {\mathcal X} : p_X(x)>0 |
\end{equation}

\noindent where $|\cdot|$ stands for the cardinality of a set. Logarithms in Eq.~(\ref{eq:entr}) to (\ref{eq:entr0}) and throughout the text are taken base two. It holds that $H_{\infty}(X) \leq H(X) \leq H_0(X)$, and the equality occurs when the outcomes in $X$ are given by a uniform distribution.

Now let $X$ and $Y$ be two jointly distributed discrete random variables taking values on alphabets $\mathcal X$ and $\mathcal Y$, respectively. The conditional entropy, min-entropy and max-entropy of $X$ given $Y$ is defined by:

\begin{equation}
H(X|Y) = \sum_{y \in {\mathcal Y}} H(X|y)
\end{equation}

\begin{equation}
H_{\infty}(X|Y) = \min_{y \in {\mathcal Y}} H_{\infty}(X|y)
\end{equation}

\begin{equation}
H_0(X|Y) = \max_{y \in {\mathcal Y}} H_0(X|y)
\end{equation}

\noindent where the entropy of a random variable given an event is the entropy of the induced random variable.

Let the state of a finite dimensional quantum system be represented by a trace one, positive semidefinite, operator on a (finite dimensional) Hilbert space $\mathcal H$. We denote by $\mathcal P(\mathcal H)$ the set of all states acting on $\mathcal H$.

Let us give some basic definitions about the quantum counterparts of these classical information measures. The equivalent of the entropy of a random variable is the von Neumann entropy of a state $\rho_X$ \cite{Vonneumann_32}. It is defined as:

\begin{equation}
H(X)_{\rho_X}  = -\mathrm{tr}(\rho_X \log \rho_X)
\end{equation}

\noindent where $\mathrm{tr}$ denotes the trace operation and we indicate with a subscript the state on which the entropy is computed. Henceforth it will be explicitly written whenever it helps clarifying a statement.

Let $\rho_{XY} \in \mathcal P(\mathcal H_{X} \otimes \mathcal H_{Y})$  be a bipartite quantum state. The conditional quantum min-entropy of $\rho_{XY}$ given $\mathcal H_{Y}$ is defined as:

\begin{equation}
\label{eq:min-quantent}
H_{\infty}({X}|{Y}) = \sup_{\sigma_{Y}}\left( -\log \min \{ \lambda | \lambda \, \textrm{id}_{X} \otimes \sigma_{Y} \geq \rho_{XY} \} \right)
\end{equation}

\noindent where $\lambda > 0$.

If $\mathcal H_{Y}$ is one dimensional:

\begin{equation}
H_{\infty}({X}|{Y}) = H_{\infty}({X}) = -\log \lambda_{\max}(\rho_X)
\end{equation}

\noindent where $\lambda_{\max}(\rho_{X})$ outputs the maximum eigenvalue of $\rho_{X}$.

We finally consider the smooth generalization of the conditional min-entropy introduced in \cite{Renner_05}. Let $\{\rho,\sigma\} \in \mathcal P(\mathcal H)$, the trace distance between $\rho$ and $\sigma$ is given by:

\begin{equation}
\frac{1}{2}||\rho-\sigma||_{1} = \mathrm{tr} \left( |\rho-\sigma| \right)
\end{equation}

The smooth entropy was first defined as an optimization over all states $\varepsilon$-close in terms of the trace distance. The smooth entropies have been redefined in terms of other measures such as the purified distance and verify additional properties \cite{Tomamichel_10, Tomamichel_11} but for the present study it suffices to consider the original definition.

Let $\rho_{XY} \in \mathcal P(\mathcal H_{X} \otimes \mathcal H_{Y})$ and $\varepsilon \geq 0$. The smooth version of Eq.~(\ref{eq:min-quantent}) is given by:

\begin{equation}
H_{\infty}^{\varepsilon}(X|Y)_{\rho_{XY}} = \sup_{\hat\rho_{XY}} H_{\infty}(X|Y)_{\hat\rho_{XY}}
\end{equation}

\noindent where the supreme is found over all $\hat\rho_{XY}$ such that $\frac{1}{2}||\rho_{XY} - \hat\rho_{XY}||_{1} \leq \varepsilon$.

\section{Information Reconciliation}

\subsection{Impact of information reconciliation on the secret key length}

One common assumption in a SKD protocol is that all the parties have access to the outcomes of an independent identically distributed experiment repeated many times. If this assumption holds the parties can safely regard an average behavior as the law of large numbers guarantees that the joint outcome will be typical with high probability. However, assuming a repetition scenario might be unrealistic in some situations, in these cases key distillation can be considered for a finite number of outcomes of a joint experiment. This second, more restrictive, scenario is sometimes referred as finite-key distillation. Both the repetition \cite{Scarani_09} and the finite-key \cite{Tomamichel_12, Hayashi_12, Salas_13} scenarios have been addressed in QKD.

The secrecy of a key $K$ can be measured in terms of its closeness to a perfect one which is uniformly random and decoupled from the eavesdropper's system $Z$. A key $K$ is considered $\varepsilon$-secure if~\cite{Konig_07}:

\begin{equation}
\frac{1}{2} \left|\left| \rho_{KZ} - \tau_K \otimes \rho_Z \right|\right|_{1} \leq \varepsilon
\end{equation}

The communications on the public channel might be one-way or two-ways. We have chosen to restrict the channel to one-way communications since our focus is on practical protocols with reduced distillation complexity, network requirements, etc. However, it should be noted that two-way communications can be used to distill a key in scenarios where one-way secret key distillation is not possible~\cite{Maurer_93} and, in general, the amount of distillable secret key with two-way communications can be strictly higher than with one-way communications~\cite{Gottesman_03, Watanabe_07}.

In the repetition scenario and aided with one-way classical communications, the maximum rate at which key can be extracted with $\varepsilon$ approaching zero as the number of repetitions goes to infinity is given by~\cite{Devetak_05}:

\begin{equation}
\label{eq:ckqkd}
K = H(X|Z) - H(X|Y)
\end{equation}

\noindent where $X$ and $Y$ are classical systems available to the legitimate parties Alice and Bob and $Z$ is a quantum system at the eavesdropper's site. 
The first term at the rhs of Eq.~(\ref{eq:ckqkd}) amounts to the randomness that can be extracted which is independent of $Z$ while the second term can be regarded as the information that Alice and Bob should exchange to reconcile $X$ and $Y$.

Eq.~(\ref{eq:ckqkd}) is valid only in the asymptotic case. However, a real system has only access to finite resources, which means that Alice and Bob not only have bounded computational power but also they have to distill a secret key from a finite number of experiments. Thus, in general there is no convergence toward an ideal key and security has to be considered for an acceptable security threshold $\varepsilon$. 

Let us assume that Alice and Bob exchange $N$ signals out of which they use $m$ for estimating their correlations and $t \leq N-m$ for key distillation. If the correlations do not verify some conditions Alice and Bob abort the protocol, $\varepsilon_\mathrm{PE}$ represents the probability that the parameter estimation procedure fails.  

Given some reconciliation protocol, $C$ stands for the set of all possible reconciliation messages and $\varepsilon_{EC}$ represents the maximum probability that the estimate at Bob's site does not coincide with Alice's string.

Let $\varepsilon_\mathrm{PA}$ represent the failure probability in the privacy amplification procedure, and $\bar\varepsilon$ be a smoothing parameter, then the rate at which the legitimate parties can distill an $\varepsilon$-secure key is bounded by~\cite{Scarani_08}:

\begin{equation}
\label{eq:flqkd}
K^{\varepsilon} \leq \frac{1}{N} \left( H_{\infty}^{\bar\varepsilon}({X^t}|{Z^NC}) - 2 \log \frac{1}{\varepsilon_\mathrm{PA}} \right)
\end{equation}

\noindent where $\varepsilon = n_\mathrm{PE} \varepsilon_\mathrm{PE} + \varepsilon_\mathrm{EC}+ \varepsilon_\mathrm{PA} +\bar\varepsilon$, and $n_\mathrm{PE}$ is the number of estimated parameters.

The smooth min-entropy in Eq.~(\ref{eq:flqkd}) can be evaluated to measure the net impact of information reconciliation~\cite{Scarani_08}:

\begin{equation}
\label{eq:hmin-sq}
H_{\infty}^{\bar\varepsilon}({X^t}|{Z^NC}) \geq H_{\infty}^{\bar\varepsilon}({X^t}|{Z^N}) - \textrm{leak}
\end{equation}

\noindent where $\textrm{leak}$ is a purely classical term that tracks the amount of information correlated with ${X^t}$ revealed during reconciliation. It is given by \cite{Renner_05}:

\begin{equation}
\textrm{leak} = H_0(C) - H_\infty(C|X^t)
\end{equation}

The main effect of an imperfect reconciliation is a reduction of the secret key rate, which in turn, in terms of the figures of merit of a QKD protocol, limits the distance range over which secret keys can be distilled~\cite{Peev_09, Scarani_09}.

\subsection{Fundamental limits of information reconciliation}
\label{sec:irec}

Let Alice and Bob be two parties holding ${x}$ and ${y}$, two $n$-length strings that are respectively $n$ outcomes of two jointly distributed random variables $X$ and $Y$. A one-way reconciliation protocol on the strings $x$ and $y$ is a protocol that produces the strings ${s_x}$ and ${s_y}$ from $x$ and $y$, respectively, after exchanging the message $c(x)$ through the public channel. 

A reconciliation protocol is considered $\varepsilon$-robust~\cite{Brassard_94} if:

\begin{equation}
\sum_{x \in \mathcal X^n, y \in \mathcal Y^n} p(x,y) p(s_x \neq s_y) \leq \varepsilon
\end{equation}

The efficiency of a reconciliation protocol can be measured using a quality parameter $\xi^\varepsilon$ that compares the amount of disclosed information with the minimum theoretical disclosure:

\begin{equation}
\xi^\varepsilon = \frac{\textrm{leak}}{nH(X|Y)}
\label{eq:efficiency}
\end{equation}

\noindent the minimum $nH(X|Y)$ is known as the Slepian-Wolf bound; it delimits the minimum rate for reliably describing a source $X$ to a distant party with access to side information $Y$~\cite{Slepian_73}.

It is well known the appropriateness of (linear) error correcting codes for the Slepian-Wolf problem~\cite{Zamir_02}. In consequence, good error correcting codes can be used for information reconciliation. Let $\mathcal{C}(n,k)$ be a linear code with coding rate $R_0=k/n$, a message of length $n-k$ called the \textit{syndrome} \footnote{Let $H$ be a parity check matrix of the code $\mathcal C$ and $x$ a vector of length $n$ the syndrome of $x$ is the vector $s(x)=H\cdot x$ of length $n-k$.} can be used to reconcile two sources with conditional entropy $nH(X|Y)$. Even if $n-k$ is greater than the theoretical minimum, for finite lengths there is always non-zero error probability. We denote the rate of decoding errors or frame error rate (FER) by the parameter $\varepsilon$. Then, a reconciliation protocol based on sending the syndrome of a linear code is $\varepsilon$-robust, and the reconciliation efficiency is given by:

\begin{equation}
\label{eq:efficiency-source-coding}
\xi^\varepsilon_{\mathcal{C}} = \frac{n-k}{n H(X|Y)} = \frac{1-R_0}{H(X|Y)}
\end{equation}

However, an acceptable FER in a communications protocol might not be sufficient in a security context. It is a common practice to divide the reconciliation process into two steps \cite{Fung_10, Tomamichel_12}. In the first one, a common string is produced, for instance using an error correcting code as we just described. In the second one, Alice uniformly at random chooses a function $f$ from a family of 2-universal hash functions \cite{Wegman_81} and computes a hash of her string $f(s_x)$. Alice sends to Bob her choice $f$ together with $f(s_x)$. Bob computes his own hash value $f(s_y)$ and the protocol aborts if $f(s_x) \neq f(s_y)$. Since the choice of the hash function is independent of $X$, only the length of the hash $\lceil - \log \varepsilon_\mathrm{EC} \rceil$ for some $\varepsilon_\mathrm{EC}>0$ is added to the leakage:

\begin{equation}
\textrm{leak}_{\mathcal C}^{\varepsilon_\mathrm{EC}} = n (1 - R_{0}) + \lceil \log \frac{1}{\varepsilon_\mathrm{EC}} \rceil
\end{equation}

The joint reconciliation process is $\varepsilon_\mathrm{EC}$-robust where $\varepsilon_{EC}$ can be chosen to be much smaller than the FER. 

It is clear from Eq.~(\ref{eq:efficiency-source-coding}) that the length of the conversation when using a code is fixed to $n-k$. That is, the amount of information does not adapt to the error rate in the channel. This is a perfect solution for the Slepian-Wolf problem since the correlations are fixed and known beforehand. However, in QKD it is common that the error rate varies from one execution to the next. In consequence, an adaptation of the coding rate is needed in order to use linear codes for reconciliation.

\section{Study of a rate-adaptive protocol}

\label{sec:rate-modulation}

In this section we study the efficiency and impact of a rate-adaptive protocol, which is in essence the $sp$-protocol in \cite{Elkouss_11} with an additional error verification step.




\subsection{Description of the rate-adaptive protocol}
\label{sec:protocol}

In the following we detail the steps of a rate-adaptive information reconciliation protocol.

{\it Step 0: Pre-conditions}. Alice and Bob agree on the following parameters: (i) a pool of shared mother codes of length $n$, constructed for different rates;  (ii) $d$ the maximum number of symbols (bits) that will be used to adapt the coding rate, and (iii) the target $\varepsilon_\mathrm{EC}$ which characterizes the length of the hashes.

{\it Step 1: Raw key exchange}. Alice and Bob obtain two correlated strings $x$ and $y$, respectively, of length $n-d$ and a precise estimate of the error rate $p_e$. If $p_e$ is outside their target rates they abort the protocol. Otherwise, both parties select the appropriate code $\mathcal{C}$ and compute the adequate number of symbols (bits) $s$  to reveal, with $s<d$, such that the coding rate is then adapted to $p_e$.

{\it Step 2: Coding}. Alice creates a extended string $\hat{x}$ 
of length $n$ by concatenating $x$ and $x'$, a uniformly random string of length $d$. Alice sends to Bob the hash value $f(\hat{x})$, the syndrome of $\hat{x}$ on $\mathcal{C}$ and the values and positions of $s$ symbols among the $d$ symbols randomly generated.

{\it Step 3: Decoding}. Bob creates a extended string of length $n$ by concatenating $y$ and $y'$, a uniformly random string of length $d$. Bob sets the values of the received $s$ symbols to their correct value. Bob computes $\hat{y}$ his estimate of $\hat{x}$ and $f(\hat{y})$ his own hash value. If $f(\hat{y}) \neq f(\hat{x})$ they abort the protocol.

We would like to remark that in {\it Step 2} both the verification tag and the reconciliation message are jointly encoded and sent to Bob. There is no extra interactivity coming from error verification, still only one message is exchanged for reconciliation and a second one from Bob to Alice is sent to notify the success or failure of the protocol.

\subsection{Leakage}
\label{sec:security}

The $sp$-protocol creates an extended system $X^tX'$ by adding $d$ symbols (bits) with random values. The Slepian-Wolf bound implies that for successful reconciliation the length of the reconciliation message should be greater than:

\begin{equation}
H(X^tX'|Y^t)=H(X^t|Y^t)+d
\end{equation}
\noindent which is trivially larger than $H(X^t|Y^t)$ if $d>0$. 

However, the appropriate comparison is in terms of the conditional smooth entropy on the reconciled system, since it is the magnitude that limits the distillable key after the reconciliation step. Lemma~\ref{lm:quantumsmoothentropy} shows that the smooth min-entropy decrease produced by the $sp$-protocol on the extended system is equivalent to the decrease produced by an error correcting code with rate R on the original system. This equivalent coding rate $R$ is given by:

\begin{equation}
\label{eq:adapted-rate}
R=\frac{k-s}{n-d}
\end{equation}

The dependence of $R$ on $d$ and $s$ allows to understand how the protocol adapts the amount of information disclosed for reconciling errors. Since the value of $d$ is fixed previous to the execution of the protocol, it is $s$, the number of symbols (bits) revealed to Bob on the public channel, the parameter available to Alice for modulating the coding rate. A higher value of $s$ increases the information available to the decoder allowing to reconcile noisier strings, while a lower value of $s$ allows to reduce the leakage by increasing the coding rate. On the other hand, $d$ sets the range of achievable rates, from $(k-d)/(n-d)$ to $k/(n-d)$. The extremal values correspond to the limiting cases of revealing the $d$ symbols (bits) and revealing no information on the public channel.

\begin{lemma}
\label{lm:quantumsmoothentropy}

Let $\rho_{X^tZ^N}$ be a bipartite state and $\sigma_{X^tX'Z^NC}$ the extension resulting from the application of the $sp$-protocol. Then the smooth min-entropy of the extended system $X^tX'$ given $Z^NC$ can be bounded by:

\begin{equation*}
H_{\infty}^{\varepsilon}(X^tX'|Z^NC)_{\sigma} \geq H_{\infty}^{\varepsilon}(X^t|Z)_{\rho} - t(1-R)- \lceil \log \frac{1}{\varepsilon_{EC}} \rceil
\end{equation*}

\end{lemma}

\begin{proof}

\begin{eqnarray}
H_{\infty}^{\varepsilon+\varepsilon'}(X^tX'|Z^NC)_{\sigma} &\geq& H_{\infty}^{\varepsilon+\varepsilon'}(X^tX'|Z^N)_{\sigma} - \textrm{leak} \nonumber \\
&=& H_{\infty}^{\varepsilon+\varepsilon'}(X^tX'|Z^NI)_{\phi} - \textrm{leak} \nonumber \\
&\geq& H_{\infty}^{\varepsilon}(X^t|Z^N)_{\phi} + H_{\infty}^{\varepsilon'}(X'|I)_{\phi} \nonumber \\
&& - \textrm{leak} \nonumber
\end{eqnarray}

Let $\varepsilon' \geq 0$. The first inequality follows from Eq.~(\ref{eq:hmin-sq}) that bounds the impact of the conversation. We can trivially extend the state on $\sigma_{X^tX'Z^N}$ to $\phi_{X^tX'Z^NI} = \sigma_{X^tX'Z^N} \otimes \textrm{id}_I$, where $I$ is a one dimensional system, without changing the value of the smooth min-entropy ($H_{\infty}^{\varepsilon + \varepsilon'}(X^tX'|Z^N)_{\sigma} = H_{\infty}^{\varepsilon + \varepsilon'} (X^tX'|Z^NI)_{\phi}$); the first equality holds by this argument. We can apply Renner's superadditivity theorem in~\cite{Renner_05} for product states to obtain the second inequality. If we now consider just the second and third terms from this last relation we obtain:

\begin{eqnarray}
H_{\infty}^{\varepsilon'}(X'|I)_{\phi} - \textrm{leak} &=& \left( s + p \right) \nonumber\\
            &&- \left( s + n(1-R_0) + \lceil \log \frac{1}{\varepsilon_{EC}} \rceil \right) \nonumber \\
 &=& - t(1-R) - \lceil \log \frac{1}{\varepsilon_{EC}} \rceil \nonumber  
\end{eqnarray}

We can choose $\varepsilon'=0$ and since $I$ is one-dimensional $H_{\infty}(X'|I)_{\phi}$ reduces to $H_{\infty}(X')_{\phi}$. Furthermore, $X'$ is classical and uniformly distributed thus maximizing the min-entropy. The leakage is obtained by tracking the amount of information sent from Alice to Bob during the protocol and subtracting the part that is independent from $X^tX'$.

We recover the desired result if we consider that $\phi_{X^tX'Z^NI}$ is also an extension of $\rho_{X^tZ^N}$ which means that $H_{\infty}^{\varepsilon} (X^t|Z^N)_{\phi} = H_{\infty}^{\varepsilon} (X^t|Z^N)_{\rho}$.

\end{proof}

\section{Simulation Results}
\label{sec:simulation-results}

In this section we compare the tradeoffs between using the $sp$-protocol, non-adapted error correcting codes and \textit{Cascade} (a well-known interactive protocol proposed in~\cite{Brassard_94} and implemented in most QKD systems). First we present the difference of the reconciliation protocols in terms of asymptotic leakage and then we plug them in a QKD protocol and compare the distillable secret key with finite resources.

The strings are assumed to be binary and are modeled as the input and output of a binary symmetric channel (BSC). This is appropriate in the case of some QKD protocols \cite{Bennett_84, Bennett_92, Scarani_04} if errors on the quantum channel are symmetric and independent.

For convenience, we have implemented the rate adaptive $sp$-protocol with irregular binary low-density parity-check (LDPC) codes since there is a wealth of material and information available: a number of matrices, decoding algorithms and communication standards have been proposed in the last years for these codes. However, non-binary LDPC codes \cite{Kasai_10} or other code families \cite{Jouguet_11}, could probably be adapted to implement the $sp$-protocol. We fixed the proportion of modulated symbols to $d/n=5\%$.

Fig.~\ref{fig:leak} shows the leakage rate ($\textrm{leak}_{\mathcal C}^{\varepsilon_\mathrm{EC}}/t$) as a function of the QBER. An optimal protocol achieving the Slepian-Wolf bound (solid line) is compared to the asymptotic $sp$-protocol computed using the theoretical analysis described in the Appendix \ref{ap:dde} (dashed line) and to \textit{Cascade} (dotted line). Note that for \textit{Cascade}, instead of upper bounding the leakage with the analytical estimate given in~\cite{Brassard_94} which might be overly pessimistic, we used as upper bound the leakage rate with large blocks of length $10^6$ (see Appendix \ref{ap:cascade} for numerical justification).

Both \textit{Cascade} and the $sp$-protocol are close to optimal for small QBERs. However, approximately over $3\%$ they begin to diverge and while the former follows closely the Slepian-Wolf bound the latter clearly has a higher leakage. 

\begin{figure}
\includegraphics[width=\linewidth]{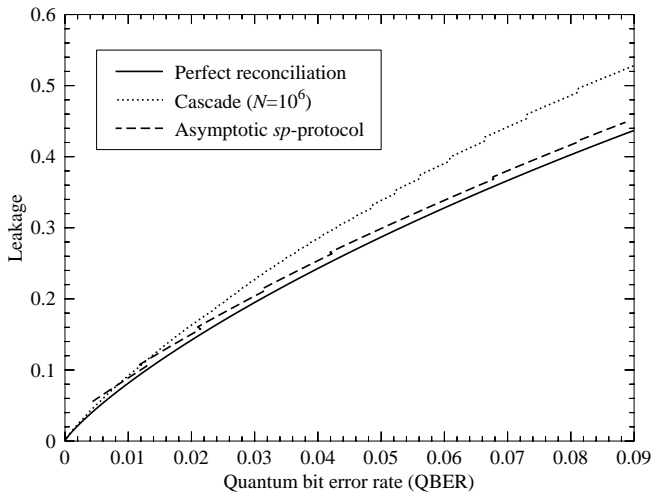}
\caption{The asymptotic leakage of the $sp$-protocol, the leakage of \textit{Cascade} and the leakage for a perfect reconciliation procedure are compared as a a function of the QBER.}
\label{fig:leak}
\end{figure}

To analyze the impact of reconciliation on the achievable secret key rate, we have chosen the prepare and measure version of BB84 and consider for simplicity and in order to highlight the effect of reconciliation, an idealized scenario: we assume that Alice and Bob have access to single photon sources and perfect detectors. Following \cite{Cai_09} the secret key in this setting can be distilled at a rate:

\begin{equation}
K^{\varepsilon} \leq \frac{t}{N} \left( \left( 1 - h(Q) \right) - \Delta(t) - \textrm{leak}/t \right)
\end{equation}

\noindent where $h$ is the binary entropy function, $Q$ is the estimated QBER that takes into account statistical fluctuations due to the finite length case, and $\Delta$ is the smoothing parameter that allows to lower bound the smooth min-entropy in Eq.~(\ref{eq:flqkd}) \cite{Scarani_08}.

\begin{figure}
\includegraphics[width=\linewidth]{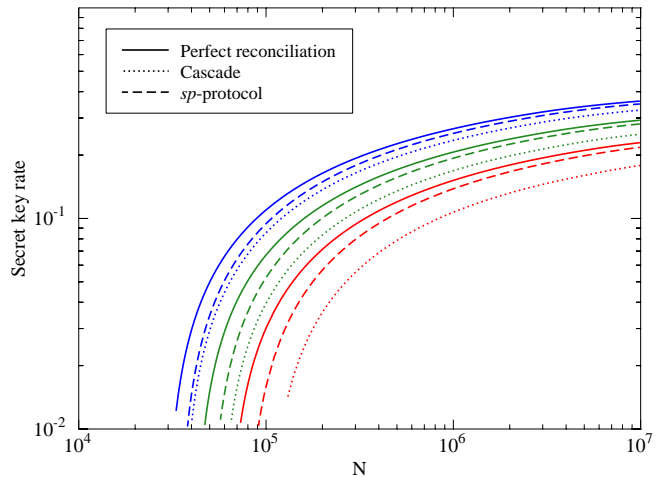}
\caption{(Color online) Secret key rate in the finite-key regime for a perfect reconciliation procedure, the $sp$-protocol, and considering the efficiency of \textit{Cascade}. Three different QBER values are considered (from left to right): $4\%$ (blue), $5\%$ (green) and $6\%$ (red). Other parameters: $\varepsilon=10^{-5}$, $\varepsilon_\mathrm{EC} = 10^{-10}$.}
\label{fig:skd}
\end{figure}

Fig.~\ref{fig:skd} shows the secret key rate as a function of the number of exchanged signals ($N$). We compare in this figure the secret key rate for three different QBER values ($4\%$, $5\%$ and $6\%$) using a perfect reconciliation protocol, \textit{Cascade}, and the $sp$-protocol. The security parameter $\varepsilon$ is set to $10^{-5}$, and $\varepsilon_\mathrm{EC} = 10^{-10}$, as suggested in \cite{Cai_09}. 

The convergence of LDPC codes towards the asymptotic value is slower than that of \textit{Cascade}  (see Appendix \ref{ap:cascade}). In consequence the optimality of the distillable key with this implementation of the $sp$-protocol increases with the length: shifting from close to \textit{Cascade} for small lengths to close to the optimal value asymptotically. For low QBERs and small lengths, the slow convergence of LDPC codes together with the good efficiency of \textit{Cascade} in this region make both secret key rates very similar. For higher QBERs, even for small lengths the LDPC implementation of the $sp$-protocol clearly outperforms \textit{Cascade}.

\section{Discussion}

This paper analyzes some improvements in the classical post-processing of QKD protocols. The key distillation process can be divided in two steps: information reconciliation and privacy amplification. Information reconciliation allows to establish a common string while in the privacy amplification step a shorter but more secure key is created. Both steps are highly coupled: in essence every bit exchanged in the information reconciliation step implies that one additional bit has to be removed of the final key in the privacy amplification step.

The problem of correcting the discrepancies between the strings of the legitimate parties is also known as the problem of source coding with side information by the information theory community. Under this paradigm, the theoretical limits of information reconciliation are given by the Slepian-Wolf bound. Information reconciliation is, then, basically error correction.

We have adopted a pragmatic approach towards error correction and used modern coding techniques well suited for QKD purposes. In a real QKD scenario we have to deal with a broad range of error rates. Further, the number of accesses to the classical public communication channel should be limited. As opposed to the eavesdropper that should, for the sake of security,  be assumed to have access to unbounded resources, the legitimate parties are equipped with a finite amount of resources.

The $sp$-protocol, induced by a mother code of rate $R_0$ allows the legitimate parties to adapt the reconciliation step to varying conditions. However, it exchanges a message longer than the optimal one. We proved that the $sp$-protocol is equivalent to the use of a code with an adapted rate $R$. The claim holds in the sense that the smooth min-entropy reduction of the former in an extended system is bounded by the reduction of the latter in the original system.

We implemented the $sp$-protocol with irregular LDPC codes. The results obtained indicate that the $sp$-protocol asymptotically behaves close to the theoretical limit. We claim no optimality in our implementation of the $sp$-protocol and certainly it could be expected that other code families are better suited to short key lengths or to other kind of correlations different than those modeled by a BSC. The analysis, however, applies to any linear error correcting code. In consequence, it allows to consider rate-adaptive information reconciliation as a specific code design problem. 
We believe that this protocol opens the doors to consider simpler and possibly better schemes for the classical postprocessing in secret key distillation protocols.

\appendix

\section{Cascade simulations}
\label{ap:cascade}

In order to estimate the asymptotic leakage of \textit{Cascade} we simulated the protocol with strings of length $10^4$, $10^5$ and $10^6$. The results on Table \ref{tab:cascadeleak} show that with a string length of $10^6$ the leakage rate has already converged.

\begingroup
\begin{table}[!ht]
    \begin{tabular}{c|c|c|c}
        QBER & $10^4$ & $10^5$ & $10^6$ \\ \hline
        0.01 & 0.0917 & 0.0914 & 0.0914 \\ 
        0.04 & 0.285 & 0.284 & 0.284 \\ 
        0.05 & 0.338 & 0.338 & 0.338 \\ 
        0.06 & 0.390 & 0.390 & 0.390 \\ 
       \hline\hline
    \end{tabular}
\caption{This table shows the leakage rate of \textit{Cascade} for strings of length $10^4$, $10^5$ and $10^6$ as a function of the QBER.}\label{tab:cascadeleak}
\end{table}
\endgroup

\section{Theoretical analysis of rate modulated codes}
\label{ap:dde}

Binary linear codes admit a bipartite graph representation in which symbols are linked with parity checks. An ensemble of irregular binary LDPC codes can be defined by the degree distributions on the edges of symbols and checks~\cite{Richardson_01}. We can study the behavior of an ensemble under a message passing algorithm by tracking the evolution of the message distributions. This recursive tracking is known as density evolution~\cite{Richardson_01} and allows to compute the asymptotic decoding threshold of a code family on a communications channel. In general, densities are updated following this recurrence relation:

\begin{equation}
p^{\ell+1}(x) = \rho \left( p_0(x) * \lambda \left( p^\ell(x) \right) \right)
\label{eq:density-evolution}
\end{equation}

\noindent where $p^{\ell}$ is the average probability on symbols on the decoding iteration $\ell$ if the code graph is tree like, $\lambda(x)$ and $\rho(x)$ are the symbol and check node degree polynomials respectively, $p_0(x)$ is the initial message density, and $*$ stands for convolution.

In section \ref{sec:simulation-results}, we focused our attention in the BSC. This channel is characterized by a single parameter: the crossover probability $\varepsilon$. That is, a bit is either noiselessly transmitted with probability $1-\varepsilon$ or flipped with probability $\varepsilon$. The channel is then modeled by the following initial density distribution:

\begin{eqnarray}
p_{0}(x) &=&  \varepsilon \Delta_{L(\varepsilon)}(x) + (1-\varepsilon) \Delta_{-L(\varepsilon)}(x) 
\label{eq:p0_ps}
\end{eqnarray}

\noindent where $L(\varepsilon) = \log \frac{\varepsilon}{1-\varepsilon}$ is a log-likelihood ratio, and $\Delta_t(x) = \delta (x-t)$ is the Dirac delta function displaced at position $t$.

Now, in the $sp$-protocol, an $n$-length raw string is composed of $n-d$ bits sent through a noisy channel, in this case the above described BSC, and $d$ bits with randomly assigned values out of which $s$ are revealed through the public and noiseless channel. Let $\sigma$ and $\pi,$ stand for the fraction of bits that are completely known and unknown to the decoder, respectively, we can compute the asymptotic behavior of the $sp$-protocol with the following initial density:

\begin{eqnarray}
p_{0}(x) &=& (1 - \pi -  \sigma) \left[ \varepsilon \Delta_{L(\varepsilon)}(x) + (1-\varepsilon) \Delta_{-L(\varepsilon)}(x) \right] \nonumber \\
 && +\pi \Delta_{0}(x) + \sigma \Delta_{\infty}(x)
\label{eq:p0_ps1}
\end{eqnarray}

\section*{Acknowledgment}

This work has been partially supported by the project Quantum Information Technologies in Madrid (QUITEMAD), reference P2009/ESP-1594 and the CHIST-ERA project Composing Quantum Channels, reference PRI-PIMCHI-2011-1071.

\bibliography{secure}

\begin{thebibliography}{37}%
\makeatletter
\providecommand \@ifxundefined [1]{%
 \@ifx{#1\undefined}
}%
\providecommand \@ifnum [1]{%
 \ifnum #1\expandafter \@firstoftwo
 \else \expandafter \@secondoftwo
 \fi
}%
\providecommand \@ifx [1]{%
 \ifx #1\expandafter \@firstoftwo
 \else \expandafter \@secondoftwo
 \fi
}%
\providecommand \natexlab [1]{#1}%
\providecommand \enquote  [1]{``#1''}%
\providecommand \bibnamefont  [1]{#1}%
\providecommand \bibfnamefont [1]{#1}%
\providecommand \citenamefont [1]{#1}%
\providecommand \href@noop [0]{\@secondoftwo}%
\providecommand \href [0]{\begingroup \@sanitize@url \@href}%
\providecommand \@href[1]{\@@startlink{#1}\@@href}%
\providecommand \@@href[1]{\endgroup#1\@@endlink}%
\providecommand \@sanitize@url [0]{\catcode `\\12\catcode `\$12\catcode
  `\&12\catcode `\#12\catcode `\^12\catcode `\_12\catcode `\%12\relax}%
\providecommand \@@startlink[1]{}%
\providecommand \@@endlink[0]{}%
\providecommand \url  [0]{\begingroup\@sanitize@url \@url }%
\providecommand \@url [1]{\endgroup\@href {#1}{\urlprefix }}%
\providecommand \urlprefix  [0]{URL }%
\providecommand \Eprint [0]{\href }%
\providecommand \doibase [0]{http://dx.doi.org/}%
\providecommand \selectlanguage [0]{\@gobble}%
\providecommand \bibinfo  [0]{\@secondoftwo}%
\providecommand \bibfield  [0]{\@secondoftwo}%
\providecommand \translation [1]{[#1]}%
\providecommand \BibitemOpen [0]{}%
\providecommand \bibitemStop [0]{}%
\providecommand \bibitemNoStop [0]{.\EOS\space}%
\providecommand \EOS [0]{\spacefactor3000\relax}%
\providecommand \BibitemShut  [1]{\csname bibitem#1\endcsname}%
\let\auto@bib@innerbib\@empty
\bibitem [{\citenamefont {Shannon}(1948)}]{Shannon_48}%
  \BibitemOpen
  \bibfield  {author} {\bibinfo {author} {\bibfnamefont {C.~E.}\ \bibnamefont
  {Shannon}},\ }\href@noop {} {\bibfield  {journal} {\bibinfo  {journal} {Bell
  Labs Tech. J.}\ }\textbf {\bibinfo {volume} {27}},\ \bibinfo {pages} {379}
  (\bibinfo {year} {1948})}\BibitemShut {NoStop}%
\bibitem [{\citenamefont {Gallager}(2001)}]{Gallager_01}%
  \BibitemOpen
  \bibfield  {author} {\bibinfo {author} {\bibfnamefont {R.~G.}\ \bibnamefont
  {Gallager}},\ }\href@noop {} {\bibfield  {journal} {\bibinfo  {journal} {IEEE
  Trans. Inf. Theory}\ }\textbf {\bibinfo {volume} {47}},\ \bibinfo {pages}
  {2681} (\bibinfo {year} {2001})}\BibitemShut {NoStop}%
\bibitem [{\citenamefont {{Mackay}}(2003)}]{MacKay_03}%
  \BibitemOpen
  \bibfield  {author} {\bibinfo {author} {\bibfnamefont {D.~J.~C.}\
  \bibnamefont {{Mackay}}},\ }\href@noop {} {\emph {\bibinfo {title}
  {{Information Theory, Inference and Learning Algorithms}}}}\ (\bibinfo
  {publisher} {Cambridge University Press},\ \bibinfo {year}
  {2003})\BibitemShut {NoStop}%
\bibitem [{\citenamefont {Shannon}(1949)}]{Shannon_49}%
  \BibitemOpen
  \bibfield  {author} {\bibinfo {author} {\bibfnamefont {C.~E.}\ \bibnamefont
  {Shannon}},\ }\href@noop {} {\bibfield  {journal} {\bibinfo  {journal} {Bell
  Labs Tech. J.}\ }\textbf {\bibinfo {volume} {28}},\ \bibinfo {pages} {656}
  (\bibinfo {year} {1949})}\BibitemShut {NoStop}%
\bibitem [{\citenamefont {Maurer}(1993)}]{Maurer_93}%
  \BibitemOpen
  \bibfield  {author} {\bibinfo {author} {\bibfnamefont {U.}~\bibnamefont
  {Maurer}},\ }\href@noop {} {\bibfield  {journal} {\bibinfo  {journal} {IEEE
  Trans. Inf. Theory}\ }\textbf {\bibinfo {volume} {39}},\ \bibinfo {pages}
  {733} (\bibinfo {year} {1993})}\BibitemShut {NoStop}%
\bibitem [{\citenamefont {G.}\ and\ \citenamefont
  {Salvail}(1994)}]{Brassard_94}%
  \BibitemOpen
  \bibfield  {author} {\bibinfo {author} {\bibfnamefont {B.}~\bibnamefont
  {G.}}\ and\ \bibinfo {author} {\bibfnamefont {L.}~\bibnamefont {Salvail}},\
  }in\ \href@noop {} {\emph {\bibinfo {booktitle} {Lecture Notes in Computer
  Science}}},\ Vol.\ \bibinfo {volume} {765}\ (\bibinfo {year} {1994})\ pp.\
  \bibinfo {pages} {410--423}\BibitemShut {NoStop}%
\bibitem [{\citenamefont {Bennett}\ \emph {et~al.}(1995)\citenamefont
  {Bennett}, \citenamefont {Brassard}, \citenamefont {Crepeau},\ and\
  \citenamefont {Maurer}}]{Bennett_95}%
  \BibitemOpen
  \bibfield  {author} {\bibinfo {author} {\bibfnamefont {C.~H.}\ \bibnamefont
  {Bennett}}, \bibinfo {author} {\bibfnamefont {G.}~\bibnamefont {Brassard}},
  \bibinfo {author} {\bibfnamefont {C.}~\bibnamefont {Crepeau}}, \ and\
  \bibinfo {author} {\bibfnamefont {U.~M.}\ \bibnamefont {Maurer}},\
  }\href@noop {} {\bibfield  {journal} {\bibinfo  {journal} {IEEE Trans. Inf.
  Theory}\ }\textbf {\bibinfo {volume} {41}},\ \bibinfo {pages} {1915}
  (\bibinfo {year} {1995})}\BibitemShut {NoStop}%
\bibitem [{\citenamefont {Bennett}\ and\ \citenamefont
  {Brassard}(1984)}]{Bennett_84}%
  \BibitemOpen
  \bibfield  {author} {\bibinfo {author} {\bibfnamefont {C.~H.}\ \bibnamefont
  {Bennett}}\ and\ \bibinfo {author} {\bibfnamefont {G.}~\bibnamefont
  {Brassard}},\ }in\ \href@noop {} {\emph {\bibinfo {booktitle} {International
  Conference on Computers, Systems and Signal Processing}}}\ (\bibinfo {year}
  {1984})\ pp.\ \bibinfo {pages} {175--179}\BibitemShut {NoStop}%
\bibitem [{\citenamefont {Gisin}\ \emph {et~al.}(2002)\citenamefont {Gisin},
  \citenamefont {Ribordy}, \citenamefont {Tittel},\ and\ \citenamefont
  {Zbinden}}]{Gisin_02}%
  \BibitemOpen
  \bibfield  {author} {\bibinfo {author} {\bibfnamefont {N.}~\bibnamefont
  {Gisin}}, \bibinfo {author} {\bibfnamefont {G.}~\bibnamefont {Ribordy}},
  \bibinfo {author} {\bibfnamefont {W.}~\bibnamefont {Tittel}}, \ and\ \bibinfo
  {author} {\bibfnamefont {H.}~\bibnamefont {Zbinden}},\ }\href@noop {}
  {\bibfield  {journal} {\bibinfo  {journal} {Rev. Mod. Phys.}\ }\textbf
  {\bibinfo {volume} {74}},\ \bibinfo {pages} {145} (\bibinfo {year}
  {2002})}\BibitemShut {NoStop}%
\bibitem [{\citenamefont {Scarani}\ \emph {et~al.}(2009)\citenamefont
  {Scarani}, \citenamefont {Bechmann-Pasquinucci}, \citenamefont {Cerf},
  \citenamefont {Du\ifmmode~\check{s}\else \v{s}\fi{}ek}, \citenamefont
  {L{\"u}tkenhaus},\ and\ \citenamefont {Peev}}]{Scarani_09}%
  \BibitemOpen
  \bibfield  {author} {\bibinfo {author} {\bibfnamefont {V.}~\bibnamefont
  {Scarani}}, \bibinfo {author} {\bibfnamefont {H.}~\bibnamefont
  {Bechmann-Pasquinucci}}, \bibinfo {author} {\bibfnamefont {N.~J.}\
  \bibnamefont {Cerf}}, \bibinfo {author} {\bibfnamefont {M.}~\bibnamefont
  {Du\ifmmode~\check{s}\else \v{s}\fi{}ek}}, \bibinfo {author} {\bibfnamefont
  {N.}~\bibnamefont {L{\"u}tkenhaus}}, \ and\ \bibinfo {author} {\bibfnamefont
  {M.}~\bibnamefont {Peev}},\ }\href@noop {} {\bibfield  {journal} {\bibinfo
  {journal} {Rev. Mod. Phys.}\ }\textbf {\bibinfo {volume} {81}},\ \bibinfo
  {pages} {1301} (\bibinfo {year} {2009})}\BibitemShut {NoStop}%
\bibitem [{\citenamefont {Elkouss}\ \emph {et~al.}(2011)\citenamefont
  {Elkouss}, \citenamefont {Martinez-Mateo},\ and\ \citenamefont
  {Martin}}]{Elkouss_11}%
  \BibitemOpen
  \bibfield  {author} {\bibinfo {author} {\bibfnamefont {D.}~\bibnamefont
  {Elkouss}}, \bibinfo {author} {\bibfnamefont {J.}~\bibnamefont
  {Martinez-Mateo}}, \ and\ \bibinfo {author} {\bibfnamefont {V.}~\bibnamefont
  {Martin}},\ }\href@noop {} {\bibfield  {journal} {\bibinfo  {journal}
  {Quantum Information and Computation}\ }\textbf {\bibinfo {volume} {11}},\
  \bibinfo {pages} {226} (\bibinfo {year} {2011})}\BibitemShut {NoStop}%
\bibitem [{\citenamefont {Elkouss}\ \emph {et~al.}(2010)\citenamefont
  {Elkouss}, \citenamefont {Martinez-Mateo},\ and\ \citenamefont
  {Martin}}]{Elkouss_10}%
  \BibitemOpen
  \bibfield  {author} {\bibinfo {author} {\bibfnamefont {D.}~\bibnamefont
  {Elkouss}}, \bibinfo {author} {\bibfnamefont {J.}~\bibnamefont
  {Martinez-Mateo}}, \ and\ \bibinfo {author} {\bibfnamefont {V.}~\bibnamefont
  {Martin}},\ }in\ \href@noop {} {\emph {\bibinfo {booktitle} {International
  Symposium on Information Theory and its Applications}}}\ (\bibinfo {year}
  {2010})\ pp.\ \bibinfo {pages} {179--184}\BibitemShut {NoStop}%
\bibitem [{\citenamefont {R\'{e}nyi}(1960)}]{Renyi_60}%
  \BibitemOpen
  \bibfield  {author} {\bibinfo {author} {\bibfnamefont {A.}~\bibnamefont
  {R\'{e}nyi}},\ }in\ \href@noop {} {\emph {\bibinfo {booktitle} {Proceedings
  of the 4th Berkeley Symposium on Mathematics, Statistics and Probability}}}\
  (\bibinfo {year} {1960})\ pp.\ \bibinfo {pages} {547--561}\BibitemShut
  {NoStop}%
\bibitem [{\citenamefont {Von~Neumann}(1932)}]{Vonneumann_32}%
  \BibitemOpen
  \bibfield  {author} {\bibinfo {author} {\bibfnamefont {J.}~\bibnamefont
  {Von~Neumann}},\ }\href@noop {} {\emph {\bibinfo {title} {Mathematische
  Grundlagen der Quantenmechanik}}}\ (\bibinfo  {publisher} {Springer},\
  \bibinfo {year} {1932})\BibitemShut {NoStop}%
\bibitem [{\citenamefont {Renner}(2005)}]{Renner_05}%
  \BibitemOpen
  \bibfield  {author} {\bibinfo {author} {\bibfnamefont {R.}~\bibnamefont
  {Renner}},\ }\href@noop {} {\bibfield  {journal} {\bibinfo  {journal} {Int.
  J. Quantum Inf.}\ }\textbf {\bibinfo {volume} {6}},\ \bibinfo {pages} {1}
  (\bibinfo {year} {2005})}\BibitemShut {NoStop}%
\bibitem [{\citenamefont {Tomamichel}\ \emph {et~al.}(2010)\citenamefont
  {Tomamichel}, \citenamefont {Colbeck},\ and\ \citenamefont
  {Renner}}]{Tomamichel_10}%
  \BibitemOpen
  \bibfield  {author} {\bibinfo {author} {\bibfnamefont {M.}~\bibnamefont
  {Tomamichel}}, \bibinfo {author} {\bibfnamefont {R.}~\bibnamefont {Colbeck}},
  \ and\ \bibinfo {author} {\bibfnamefont {R.}~\bibnamefont {Renner}},\
  }\href@noop {} {\bibfield  {journal} {\bibinfo  {journal} {IEEE Trans. Inf.
  Theory}\ }\textbf {\bibinfo {volume} {56}},\ \bibinfo {pages} {4674}
  (\bibinfo {year} {2010})}\BibitemShut {NoStop}%
\bibitem [{\citenamefont {Tomamichel}\ and\ \citenamefont
  {Renner}(2011)}]{Tomamichel_11}%
  \BibitemOpen
  \bibfield  {author} {\bibinfo {author} {\bibfnamefont {M.}~\bibnamefont
  {Tomamichel}}\ and\ \bibinfo {author} {\bibfnamefont {R.}~\bibnamefont
  {Renner}},\ }\href@noop {} {\bibfield  {journal} {\bibinfo  {journal} {Phys.
  Rev. Lett.}\ }\textbf {\bibinfo {volume} {106}},\ \bibinfo {pages} {110506}
  (\bibinfo {year} {2011})}\BibitemShut {NoStop}%
\bibitem [{\citenamefont {Tomamichel}\ \emph {et~al.}(2012)\citenamefont
  {Tomamichel}, \citenamefont {Lim}, \citenamefont {Gisin},\ and\ \citenamefont
  {Renner}}]{Tomamichel_12}%
  \BibitemOpen
  \bibfield  {author} {\bibinfo {author} {\bibfnamefont {M.}~\bibnamefont
  {Tomamichel}}, \bibinfo {author} {\bibfnamefont {C.~C.~W.}\ \bibnamefont
  {Lim}}, \bibinfo {author} {\bibfnamefont {N.}~\bibnamefont {Gisin}}, \ and\
  \bibinfo {author} {\bibfnamefont {R.}~\bibnamefont {Renner}},\ }\href@noop {}
  {\bibfield  {journal} {\bibinfo  {journal} {Nat. Commun.}\ }\textbf {\bibinfo
  {volume} {3}},\ \bibinfo {pages} {1} (\bibinfo {year} {2012})}\BibitemShut
  {NoStop}%
\bibitem [{\citenamefont {Hayashi}\ and\ \citenamefont
  {Tsurumaru}(2012)}]{Hayashi_12}%
  \BibitemOpen
  \bibfield  {author} {\bibinfo {author} {\bibfnamefont {M.}~\bibnamefont
  {Hayashi}}\ and\ \bibinfo {author} {\bibfnamefont {T.}~\bibnamefont
  {Tsurumaru}},\ }\href@noop {} {\bibfield  {journal} {\bibinfo  {journal} {New
  J. Phys.}\ }\textbf {\bibinfo {volume} {14}},\ \bibinfo {pages} {093014}
  (\bibinfo {year} {2012})}\BibitemShut {NoStop}%
\bibitem [{\citenamefont {Salas}(2013)}]{Salas_13}%
  \BibitemOpen
  \bibfield  {author} {\bibinfo {author} {\bibfnamefont {P.~J.}\ \bibnamefont
  {Salas}},\ }\href@noop {} {\bibfield  {journal} {\bibinfo  {journal} {Quantum
  Information and Computation}\ }\textbf {\bibinfo {volume} {13}},\ \bibinfo
  {pages} {861} (\bibinfo {year} {2013})}\BibitemShut {NoStop}%
\bibitem [{\citenamefont {K\"onig}\ \emph {et~al.}(2007)\citenamefont
  {K\"onig}, \citenamefont {Renner}, \citenamefont {Bariska},\ and\
  \citenamefont {Maurer}}]{Konig_07}%
  \BibitemOpen
  \bibfield  {author} {\bibinfo {author} {\bibfnamefont {R.}~\bibnamefont
  {K\"onig}}, \bibinfo {author} {\bibfnamefont {R.}~\bibnamefont {Renner}},
  \bibinfo {author} {\bibfnamefont {A.}~\bibnamefont {Bariska}}, \ and\
  \bibinfo {author} {\bibfnamefont {U.}~\bibnamefont {Maurer}},\ }\href@noop {}
  {\bibfield  {journal} {\bibinfo  {journal} {Phys. Rev. Lett.}\ }\textbf
  {\bibinfo {volume} {98}},\ \bibinfo {pages} {140502} (\bibinfo {year}
  {2007})}\BibitemShut {NoStop}%
\bibitem [{\citenamefont {Gottesman}\ and\ \citenamefont
  {Lo}(2003)}]{Gottesman_03}%
  \BibitemOpen
  \bibfield  {author} {\bibinfo {author} {\bibfnamefont {D.}~\bibnamefont
  {Gottesman}}\ and\ \bibinfo {author} {\bibfnamefont {H.~K.}\ \bibnamefont
  {Lo}},\ }\href@noop {} {\bibfield  {journal} {\bibinfo  {journal} {IEEE
  Trans. Inf. Theory}\ }\textbf {\bibinfo {volume} {49}},\ \bibinfo {pages}
  {457} (\bibinfo {year} {2003})}\BibitemShut {NoStop}%
\bibitem [{\citenamefont {Watanabe}\ \emph {et~al.}(2007)\citenamefont
  {Watanabe}, \citenamefont {Matsumoto}, \citenamefont {Uyematsu},\ and\
  \citenamefont {Kawano}}]{Watanabe_07}%
  \BibitemOpen
  \bibfield  {author} {\bibinfo {author} {\bibfnamefont {S.}~\bibnamefont
  {Watanabe}}, \bibinfo {author} {\bibfnamefont {R.}~\bibnamefont {Matsumoto}},
  \bibinfo {author} {\bibfnamefont {T.}~\bibnamefont {Uyematsu}}, \ and\
  \bibinfo {author} {\bibfnamefont {Y.}~\bibnamefont {Kawano}},\ }\href@noop {}
  {\bibfield  {journal} {\bibinfo  {journal} {Phys. Rev. A}\ }\textbf {\bibinfo
  {volume} {76}},\ \bibinfo {pages} {032312} (\bibinfo {year}
  {2007})}\BibitemShut {NoStop}%
\bibitem [{\citenamefont {Devetak}\ and\ \citenamefont
  {Winter}(2005)}]{Devetak_05}%
  \BibitemOpen
  \bibfield  {author} {\bibinfo {author} {\bibfnamefont {I.}~\bibnamefont
  {Devetak}}\ and\ \bibinfo {author} {\bibfnamefont {A.}~\bibnamefont
  {Winter}},\ }\href@noop {} {\bibfield  {journal} {\bibinfo  {journal}
  {Proceedings of the Royal Society A: Mathematical, Physical and Engineering
  Sciences}\ }\textbf {\bibinfo {volume} {461}},\ \bibinfo {pages} {207}
  (\bibinfo {year} {2005})}\BibitemShut {NoStop}%
\bibitem [{\citenamefont {Scarani}\ and\ \citenamefont
  {Renner}(2008)}]{Scarani_08}%
  \BibitemOpen
  \bibfield  {author} {\bibinfo {author} {\bibfnamefont {V.}~\bibnamefont
  {Scarani}}\ and\ \bibinfo {author} {\bibfnamefont {R.}~\bibnamefont
  {Renner}},\ }\href@noop {} {\bibfield  {journal} {\bibinfo  {journal} {Phys.
  Rev. Lett.}\ }\textbf {\bibinfo {volume} {100}},\ \bibinfo {pages} {200501}
  (\bibinfo {year} {2008})}\BibitemShut {NoStop}%
\bibitem [{\citenamefont {Peev}\ \emph {et~al.}(2009)\citenamefont {Peev},
  \citenamefont {Pacher}, \citenamefont {All\'{e}aume}, \citenamefont
  {Barreiro}, \citenamefont {Bouda}, \citenamefont {Boxleitner}, \citenamefont
  {Debuisschert}, \citenamefont {Diamanti}, \citenamefont {Dianati},
  \citenamefont {Dynes}, \citenamefont {Fasel}, \citenamefont {Fossier},
  \citenamefont {F\"{u}rst}, \citenamefont {Gautier}, \citenamefont {Gay},
  \citenamefont {Gisin}, \citenamefont {Grangier}, \citenamefont {Happe},
  \citenamefont {Hasani}, \citenamefont {Hentschel}, \citenamefont {H\"{u}bel},
  \citenamefont {Humer}, \citenamefont {L\"{a}nger}, \citenamefont {Legr\'{e}},
  \citenamefont {Lieger}, \citenamefont {Lodewyck}, \citenamefont
  {Lor\"{u}nser}, \citenamefont {L\"{u}tkenhaus}, \citenamefont {Marhold},
  \citenamefont {Matyus}, \citenamefont {Maurhart}, \citenamefont {Monat},
  \citenamefont {Nauerth}, \citenamefont {Page}, \citenamefont {Poppe},
  \citenamefont {Querasser}, \citenamefont {Ribordy}, \citenamefont {Robyr},
  \citenamefont {Salvail}, \citenamefont {Sharpe}, \citenamefont {Shields},
  \citenamefont {Stucki}, \citenamefont {Suda}, \citenamefont {Tamas},
  \citenamefont {Themel}, \citenamefont {Thew}, \citenamefont {Thoma},
  \citenamefont {Treiber}, \citenamefont {Trinkler}, \citenamefont
  {Tualle-Brouri}, \citenamefont {Vannel}, \citenamefont {Walenta},
  \citenamefont {Weier}, \citenamefont {Weinfurter}, \citenamefont {Wimberger},
  \citenamefont {Yuan}, \citenamefont {Zbinden},\ and\ \citenamefont
  {Zeilinger}}]{Peev_09}%
  \BibitemOpen
  \bibfield  {author} {\bibinfo {author} {\bibfnamefont {M.}~\bibnamefont
  {Peev}}, \bibinfo {author} {\bibfnamefont {C.}~\bibnamefont {Pacher}},
  \bibinfo {author} {\bibfnamefont {R.}~\bibnamefont {All\'{e}aume}}, \bibinfo
  {author} {\bibfnamefont {C.}~\bibnamefont {Barreiro}}, \bibinfo {author}
  {\bibfnamefont {J.}~\bibnamefont {Bouda}}, \bibinfo {author} {\bibfnamefont
  {W.}~\bibnamefont {Boxleitner}}, \bibinfo {author} {\bibfnamefont
  {T.}~\bibnamefont {Debuisschert}}, \bibinfo {author} {\bibfnamefont
  {E.}~\bibnamefont {Diamanti}}, \bibinfo {author} {\bibfnamefont
  {M.}~\bibnamefont {Dianati}}, \bibinfo {author} {\bibfnamefont {J.~F.}\
  \bibnamefont {Dynes}}, \bibinfo {author} {\bibfnamefont {S.}~\bibnamefont
  {Fasel}}, \bibinfo {author} {\bibfnamefont {S.}~\bibnamefont {Fossier}},
  \bibinfo {author} {\bibfnamefont {M.}~\bibnamefont {F\"{u}rst}}, \bibinfo
  {author} {\bibfnamefont {J.-D.}\ \bibnamefont {Gautier}}, \bibinfo {author}
  {\bibfnamefont {O.}~\bibnamefont {Gay}}, \bibinfo {author} {\bibfnamefont
  {N.}~\bibnamefont {Gisin}}, \bibinfo {author} {\bibfnamefont
  {P.}~\bibnamefont {Grangier}}, \bibinfo {author} {\bibfnamefont
  {A.}~\bibnamefont {Happe}}, \bibinfo {author} {\bibfnamefont
  {Y.}~\bibnamefont {Hasani}}, \bibinfo {author} {\bibfnamefont
  {M.}~\bibnamefont {Hentschel}}, \bibinfo {author} {\bibfnamefont
  {H.}~\bibnamefont {H\"{u}bel}}, \bibinfo {author} {\bibfnamefont
  {G.}~\bibnamefont {Humer}}, \bibinfo {author} {\bibfnamefont
  {T.}~\bibnamefont {L\"{a}nger}}, \bibinfo {author} {\bibfnamefont
  {M.}~\bibnamefont {Legr\'{e}}}, \bibinfo {author} {\bibfnamefont
  {R.}~\bibnamefont {Lieger}}, \bibinfo {author} {\bibfnamefont
  {J.}~\bibnamefont {Lodewyck}}, \bibinfo {author} {\bibfnamefont
  {T.}~\bibnamefont {Lor\"{u}nser}}, \bibinfo {author} {\bibfnamefont
  {N.}~\bibnamefont {L\"{u}tkenhaus}}, \bibinfo {author} {\bibfnamefont
  {A.}~\bibnamefont {Marhold}}, \bibinfo {author} {\bibfnamefont
  {T.}~\bibnamefont {Matyus}}, \bibinfo {author} {\bibfnamefont
  {O.}~\bibnamefont {Maurhart}}, \bibinfo {author} {\bibfnamefont
  {L.}~\bibnamefont {Monat}}, \bibinfo {author} {\bibfnamefont
  {S.}~\bibnamefont {Nauerth}}, \bibinfo {author} {\bibfnamefont {J.-B.}\
  \bibnamefont {Page}}, \bibinfo {author} {\bibfnamefont {A.}~\bibnamefont
  {Poppe}}, \bibinfo {author} {\bibfnamefont {E.}~\bibnamefont {Querasser}},
  \bibinfo {author} {\bibfnamefont {G.}~\bibnamefont {Ribordy}}, \bibinfo
  {author} {\bibfnamefont {S.}~\bibnamefont {Robyr}}, \bibinfo {author}
  {\bibfnamefont {L.}~\bibnamefont {Salvail}}, \bibinfo {author} {\bibfnamefont
  {A.~W.}\ \bibnamefont {Sharpe}}, \bibinfo {author} {\bibfnamefont {A.~J.}\
  \bibnamefont {Shields}}, \bibinfo {author} {\bibfnamefont {D.}~\bibnamefont
  {Stucki}}, \bibinfo {author} {\bibfnamefont {M.}~\bibnamefont {Suda}},
  \bibinfo {author} {\bibfnamefont {C.}~\bibnamefont {Tamas}}, \bibinfo
  {author} {\bibfnamefont {T.}~\bibnamefont {Themel}}, \bibinfo {author}
  {\bibfnamefont {R.~T.}\ \bibnamefont {Thew}}, \bibinfo {author}
  {\bibfnamefont {Y.}~\bibnamefont {Thoma}}, \bibinfo {author} {\bibfnamefont
  {A.}~\bibnamefont {Treiber}}, \bibinfo {author} {\bibfnamefont
  {P.}~\bibnamefont {Trinkler}}, \bibinfo {author} {\bibfnamefont
  {R.}~\bibnamefont {Tualle-Brouri}}, \bibinfo {author} {\bibfnamefont
  {F.}~\bibnamefont {Vannel}}, \bibinfo {author} {\bibfnamefont
  {N.}~\bibnamefont {Walenta}}, \bibinfo {author} {\bibfnamefont
  {H.}~\bibnamefont {Weier}}, \bibinfo {author} {\bibfnamefont
  {H.}~\bibnamefont {Weinfurter}}, \bibinfo {author} {\bibfnamefont
  {I.}~\bibnamefont {Wimberger}}, \bibinfo {author} {\bibfnamefont {Z.~L.}\
  \bibnamefont {Yuan}}, \bibinfo {author} {\bibfnamefont {H.}~\bibnamefont
  {Zbinden}}, \ and\ \bibinfo {author} {\bibfnamefont {A.}~\bibnamefont
  {Zeilinger}},\ }\href@noop {} {\bibfield  {journal} {\bibinfo  {journal} {New
  J. Phys.}\ }\textbf {\bibinfo {volume} {11}},\ \bibinfo {pages} {075001}
  (\bibinfo {year} {2009})}\BibitemShut {NoStop}%
\bibitem [{\citenamefont {Slepian}\ and\ \citenamefont
  {Wolf}(1973)}]{Slepian_73}%
  \BibitemOpen
  \bibfield  {author} {\bibinfo {author} {\bibfnamefont {D.}~\bibnamefont
  {Slepian}}\ and\ \bibinfo {author} {\bibfnamefont {J.}~\bibnamefont {Wolf}},\
  }\href@noop {} {\bibfield  {journal} {\bibinfo  {journal} {IEEE Trans. Inf.
  Theory,}\ }\textbf {\bibinfo {volume} {19}},\ \bibinfo {pages} {471}
  (\bibinfo {year} {1973})}\BibitemShut {NoStop}%
\bibitem [{\citenamefont {Zamir}\ \emph {et~al.}(2002)\citenamefont {Zamir},
  \citenamefont {Shamai},\ and\ \citenamefont {Erez}}]{Zamir_02}%
  \BibitemOpen
  \bibfield  {author} {\bibinfo {author} {\bibfnamefont {R.}~\bibnamefont
  {Zamir}}, \bibinfo {author} {\bibfnamefont {S.}~\bibnamefont {Shamai}}, \
  and\ \bibinfo {author} {\bibfnamefont {U.}~\bibnamefont {Erez}},\ }\href@noop
  {} {\bibfield  {journal} {\bibinfo  {journal} {IEEE Trans. Inf. Theory}\
  }\textbf {\bibinfo {volume} {48}},\ \bibinfo {pages} {1250} (\bibinfo {year}
  {2002})}\BibitemShut {NoStop}%
\bibitem [{Note1()}]{Note1}%
  \BibitemOpen
  \bibinfo {note} {Let $H$ be a parity check matrix of the code $\protect
  \mathcal C$ and $x$ a vector of length $n$ the syndrome of $x$ is the vector
  $s(x)=H\cdot x$ of length $n-k$.}\BibitemShut {Stop}%
\bibitem [{\citenamefont {Fung}\ \emph {et~al.}(2010)\citenamefont {Fung},
  \citenamefont {Ma},\ and\ \citenamefont {Chau}}]{Fung_10}%
  \BibitemOpen
  \bibfield  {author} {\bibinfo {author} {\bibfnamefont {C.-H.~F.}\
  \bibnamefont {Fung}}, \bibinfo {author} {\bibfnamefont {X.}~\bibnamefont
  {Ma}}, \ and\ \bibinfo {author} {\bibfnamefont {H.~F.}\ \bibnamefont
  {Chau}},\ }\href@noop {} {\bibfield  {journal} {\bibinfo  {journal} {Phys.
  Rev. A}\ }\textbf {\bibinfo {volume} {81}},\ \bibinfo {pages} {012318}
  (\bibinfo {year} {2010})}\BibitemShut {NoStop}%
\bibitem [{\citenamefont {Wegman}\ and\ \citenamefont
  {Carter}(1981)}]{Wegman_81}%
  \BibitemOpen
  \bibfield  {author} {\bibinfo {author} {\bibfnamefont {M.~N.}\ \bibnamefont
  {Wegman}}\ and\ \bibinfo {author} {\bibfnamefont {L.}~\bibnamefont
  {Carter}},\ }\href@noop {} {\bibfield  {journal} {\bibinfo  {journal} {J.
  Comput. Syst. Sci.}\ }\textbf {\bibinfo {volume} {22}},\ \bibinfo {pages}
  {265} (\bibinfo {year} {1981})}\BibitemShut {NoStop}%
\bibitem [{\citenamefont {Bennett}\ \emph {et~al.}(1992)\citenamefont
  {Bennett}, \citenamefont {Bessette}, \citenamefont {Brassard}, \citenamefont
  {Salvail},\ and\ \citenamefont {Smolin}}]{Bennett_92}%
  \BibitemOpen
  \bibfield  {author} {\bibinfo {author} {\bibfnamefont {C.~H.}\ \bibnamefont
  {Bennett}}, \bibinfo {author} {\bibfnamefont {F.}~\bibnamefont {Bessette}},
  \bibinfo {author} {\bibfnamefont {G.}~\bibnamefont {Brassard}}, \bibinfo
  {author} {\bibfnamefont {L.}~\bibnamefont {Salvail}}, \ and\ \bibinfo
  {author} {\bibfnamefont {J.}~\bibnamefont {Smolin}},\ }\href@noop {}
  {\bibfield  {journal} {\bibinfo  {journal} {J. Cryptology}\ }\textbf
  {\bibinfo {volume} {5}},\ \bibinfo {pages} {3} (\bibinfo {year}
  {1992})}\BibitemShut {NoStop}%
\bibitem [{\citenamefont {Scarani}\ \emph {et~al.}(2004)\citenamefont
  {Scarani}, \citenamefont {Ac\'{i}n}, \citenamefont {Ribordy},\ and\
  \citenamefont {Gisin}}]{Scarani_04}%
  \BibitemOpen
  \bibfield  {author} {\bibinfo {author} {\bibfnamefont {V.}~\bibnamefont
  {Scarani}}, \bibinfo {author} {\bibfnamefont {A.}~\bibnamefont {Ac\'{i}n}},
  \bibinfo {author} {\bibfnamefont {G.}~\bibnamefont {Ribordy}}, \ and\
  \bibinfo {author} {\bibfnamefont {N.}~\bibnamefont {Gisin}},\ }\href@noop {}
  {\bibfield  {journal} {\bibinfo  {journal} {Phys. Rev. Lett.}\ }\textbf
  {\bibinfo {volume} {92}},\ \bibinfo {pages} {057901} (\bibinfo {year}
  {2004})}\BibitemShut {NoStop}%
\bibitem [{\citenamefont {Kasai}\ \emph {et~al.}(2010)\citenamefont {Kasai},
  \citenamefont {Matsumoto},\ and\ \citenamefont {Sakaniwa}}]{Kasai_10}%
  \BibitemOpen
  \bibfield  {author} {\bibinfo {author} {\bibfnamefont {K.}~\bibnamefont
  {Kasai}}, \bibinfo {author} {\bibfnamefont {R.}~\bibnamefont {Matsumoto}}, \
  and\ \bibinfo {author} {\bibfnamefont {K.}~\bibnamefont {Sakaniwa}},\ }in\
  \href@noop {} {\emph {\bibinfo {booktitle} {International Symposium on
  Information Theory and its Applications}}}\ (\bibinfo {year} {2010})\ pp.\
  \bibinfo {pages} {922--927}\BibitemShut {NoStop}%
\bibitem [{\citenamefont {Jouguet}\ \emph {et~al.}(2011)\citenamefont
  {Jouguet}, \citenamefont {Kunz-Jacques},\ and\ \citenamefont
  {Leverrier}}]{Jouguet_11}%
  \BibitemOpen
  \bibfield  {author} {\bibinfo {author} {\bibfnamefont {P.}~\bibnamefont
  {Jouguet}}, \bibinfo {author} {\bibfnamefont {S.}~\bibnamefont
  {Kunz-Jacques}}, \ and\ \bibinfo {author} {\bibfnamefont {A.}~\bibnamefont
  {Leverrier}},\ }\href@noop {} {\bibfield  {journal} {\bibinfo  {journal}
  {Phys. Rev. A}\ }\textbf {\bibinfo {volume} {84}},\ \bibinfo {pages} {062317}
  (\bibinfo {year} {2011})}\BibitemShut {NoStop}%
\bibitem [{\citenamefont {Cai}\ and\ \citenamefont {Scarani}(2009)}]{Cai_09}%
  \BibitemOpen
  \bibfield  {author} {\bibinfo {author} {\bibfnamefont {R.~Y.~Q.}\
  \bibnamefont {Cai}}\ and\ \bibinfo {author} {\bibfnamefont {V.}~\bibnamefont
  {Scarani}},\ }\href@noop {} {\bibfield  {journal} {\bibinfo  {journal} {New
  J. Phys.}\ }\textbf {\bibinfo {volume} {11}},\ \bibinfo {pages} {045024}
  (\bibinfo {year} {2009})}\BibitemShut {NoStop}%
\bibitem [{\citenamefont {Richardson}\ and\ \citenamefont
  {Urbanke}(2001)}]{Richardson_01}%
  \BibitemOpen
  \bibfield  {author} {\bibinfo {author} {\bibfnamefont {T.}~\bibnamefont
  {Richardson}}\ and\ \bibinfo {author} {\bibfnamefont {R.}~\bibnamefont
  {Urbanke}},\ }\href@noop {} {\bibfield  {journal} {\bibinfo  {journal} {IEEE
  Trans. Inf. Theory}\ }\textbf {\bibinfo {volume} {47}},\ \bibinfo {pages}
  {599} (\bibinfo {year} {2001})}\BibitemShut {NoStop}%
\end{thebibliography}%

\end{document}